\documentclass[10pt,conference]{IEEEtran}
\usepackage[fleqn]{amsmath}
\usepackage[pdftex]{graphicx}
\usepackage{amssymb}
\usepackage{amsthm}
\usepackage{epstopdf}
\newtheorem{theorem}{Theorem}

\title{  AWGN Channel Capacity of Energy Harvesting Transmitters with a Finite Energy Buffer}
 
\author{\IEEEauthorblockN{Deekshith P K, Vinod Sharma}
\IEEEauthorblockA{ Dept. of ECE, Bangalore, India\\
Email: \{deeks,vinod\}@ece.iisc.ernet.in}
\and
\IEEEauthorblockN{ R Rajesh }
\IEEEauthorblockA{CABS, DRDO, Bangalore, India\\
Email: rajesh81r@gmail.com}}
\begin{document}
\maketitle
\begin{abstract}
We consider an AWGN channel with a transmitter powered by an energy harvesting source. The node is equipped with a finite energy buffer. Such a system can be modelled as a channel with side information (about energy in the energy buffer) causally known at the transmitter. The receiver may or may not have the side information. We prove that Markov energy management policies are sufficient to achieve the capacity of the system and provide a single letter characterization for the capacity. The computation of the capacity is expensive. Therefore, we discuss an achievable scheme that is easy to compute. This achievable rate converges to the infinite buffer capacity as the buffer length increases.
\end{abstract}
\noindent
\section{Introduction}
 Analysis of various channel models with the transmitters and the receivers being powered by energy harvesting sources has received considerable attention in the information theoretic literature. Of late, there has been  interest in the analysis of single and multi user channels with finite energy buffers (\cite{Rajesh_journey}). This model stands in between the two extremes of the infinite buffer and the no buffer model which have been well studied (\cite{Rajesh1},~\cite{Ulukus2}). However, the capacity for a finite buffer model is not known yet.

The capacity of an AWGN channel having an infinite buffer to store energy  is provided in \cite{Rajesh1}, \cite{Ulukus1}. \cite{Rajesh1} also provides the capacity of a system with no storage buffer and full causal knowledge of the energy harvesting process at the receiver and the capacity when energy is spent in processes other than transmission, i.e., sensing, processing etc. Further, achievable rates for models with storage inefficiencies taken into account are also obtained. When the knowledge of the energy harvesting process is not known at the receiver, the channel is provided in \cite{Ulukus2}.
 
 A more realistic model is to consider a transmitter with a finite battery capacity. Compared to the infinite buffer and no buffer model, such a channel model has not been much studied in literature. However, various schemes providing achievable rates are known for the finite buffer model. \cite{Rajesh_journey} obtains achievable rates via Markov policies for an AWGN channel using stochastic approximation. This has been extended to a Gaussian multiple access channel (GMAC) model in \cite{Rajesh_ita}. \cite{Rajesh_ita} also provides the capacity region of a GMAC with the transmitting nodes having infinite energy storage capability. This result is also provided independently in \cite{drajan}. \cite{Rajesh_ita} also considers the case when the information about the energy is not known at the receiver.  \cite{Rajesh2} finds the Shannon capacity of a fading AWGN channel with transmitters having perfect/no  information of the fading process at the transmitter. \cite{Rajesh_q} combines the information theoretic and queueing theoretic approaches to obtain the Shannon capacity of an AWGN channel with energy harvesting transmitters having a data queue. In \cite{Rajesh2} and \cite{Rajesh_q} the basic model has been extended to non-ideal models, incorporating battery leakage and energy storage inefficiencies as well. \cite{O_Sim} studies the problem of communicating a measurement to a receiver by optimally allocating energy between the source and the channel encoder. An energy management scheme for systems with finite energy and data buffer capacity is provided in \cite{Sri_Koksal}. \cite{Biyikoglu} provides a concise survey on the queueing theoretic and information theoretic results in the literature on communication systems with energy harvesting sources.
   
In this paper, we find the information theoretic capacity of an AWGN channel when the energy harvesting transmitter node has  only a finite energy buffer. We provide a single letter characterization for the system capacity when there is energy buffer state information available at the receiver (BSIR). The capacity expression for the channel with no BSIR is in terms of ``Shannon strategies". Also, we prove that the system capacity is achieved by the class of Markov energy management policies in the system with BSIR. We also prove the convergence of the capacity of the finite buffer case to the capacity of the infinite buffer system as the buffer length tends to infinity. 

The paper is organized as follows. In Section II we present the system model and explain the notations used. In Section III we obtain the expression for the capacity of a finite buffer system with BSIR and prove that the capacity is achieved by the class of Markov policies. In Section IV, we present achievable rate via truncated Gaussian signalling and a greedy policy. Section V provides  the capacity with no BSIR in terms of ``Shannon strategies". Section VI concludes the paper.

\section{Model and Notation}
\begin{figure}[h]
\includegraphics[scale=.45]{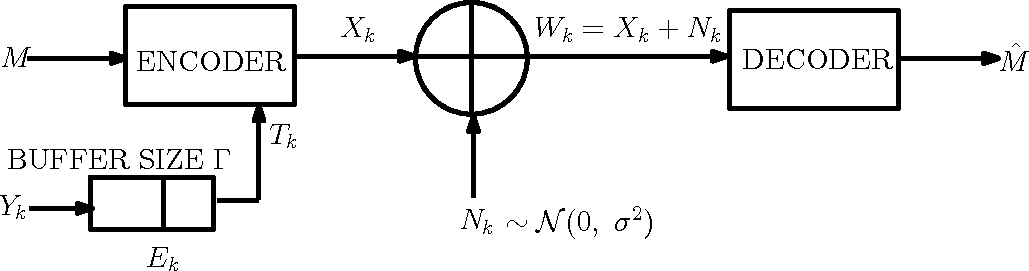}
\caption{The Model}
\label{cha_model}
\end{figure}
We consider an AWGN channel with the transmitter powered by an energy harvesting source (Fig \ref{cha_model}). The transmitting node is equipped with an energy storage device of finite storage capacity $\Gamma.$ The system is slotted. Let $Y_k$ be the energy harvested in slot $k-1$ (which is available for use in slot $k$). We assume $\{Y_k\}$ to be independent, identically distributed (i.i.d) with $\mathbb{E}[Y_k^2] < \infty$. Let $E_k$ denote the energy available in the buffer at the beginning of time slot $k$ (not including $Y_k$; hence energy available for transmission at time $k$ is $\hat{E}_k = Y_k+E_k$). $\{N_k\}$ denotes the additive i.i.d Gaussian noise with mean 0 and variance $\sigma^2$ $\left(\mathcal{N}(0,\sigma^2)\right)$. If $X_k$ is the channel input at time $k$, then the channel output $W_k=X_k+N_k$. We denote the energy used by the transmitting node at time $k$ as $T_k =X_k^2 \leq E_k+Y_k$. We assume that the energy is being used only for transmission.   This is a commonly used assumption in this literature (\cite{Rajesh1},~\cite{Ulukus1}).

$\hat{E}_k$, the energy available at time $k$ denotes the state of the system at time $k$. We consider Markovian energy management policies, i.e., the energy $T_k $ used by the transmitting node  at time $k$ depends only on $\hat{E}_k$. We will show that such policies are sufficient to achieve capacity. Since $\mathbb{E}[Y_k^2]< \infty$,  $\mathbb{E}[T_k]=\mathbb{E}[X_k^2]<\infty$ for each $k$.  

The $\{E_k\}$ process evolves as
\begin{align}
\label{buf_update}
\hspace{40pt}E_{k+1}=\min\{\Gamma,~\hat{E}_{k+1}-T_{k+1}\}.
\end{align}
For a Markovian policy, $\{E_k\}$, $\{\hat{E}_k\}$ and $\{(E_k,X_k)\}$ are Markov chains. 

If the energy is measured and used in quanta, then the energy buffer $\Gamma$ can store a finite number of quanta. Also, then we take $\{T_k\}$, $\{Y_k\}$ and $\{\hat{E}_k\}$ as discrete in number of energy quanta. We allow $T_k$, $Y_k$ and $\hat{E}_k$ to take countably infinite number of values, although $\{E_k\}$ has a finite state space.

When $\{E_k\}$ is a finite state Markov chain then if it is irreducible, it becomes a positive recurrent (ergodic) Markov chain with a unique stationary distribution. If $\{E_k\}$ is not irreducible (it depends on the energy management policy chosen), then, its state space can be decomposed into a finite number of ergodic closed classes, and there will not be a unique stationary distribution. However, 
\begin{align}
\label{pro_lim}
\hspace{60pt}\lim_{n \rightarrow \infty} \frac{1}{n} \sum_{k=1}^{n} P(E_k \in A)
\end{align}
always exists for each measurable set $A$ although in the non irreducible case, the limiting distribution will depend on the initial distribution. This implies that for any policy and any initial distribution, $\{E_k\}$ is an asymptotically mean stationary (AMS) sequence. Then $\{\hat{E}_k\}$ and $\{(\hat{E}_k,E_k,Y_k,X_k)\}$ is also an AMS sequence. It will be ergodic when $\{E_k\}$ is irreducible.

Another case is when $E_k \in [0,\Gamma]$ has an uncountable state space. Now the state space can be taken compact (e.g., $[0,\Gamma]$ itself). If the energy management policy is such that $s \mapsto P_{X(s)}$ (where $P_{X(s)}$ denotes the distribution of $X_k$ when $\hat{E}_k = s$) is continuous in the weak topology of $P_X$, then $\{E_k\}$ is a Feller continuous Markov chain (\cite{Athreya}). Also, because the state space is compact, (\ref{pro_lim}) has a limit point. Then from \cite{Athreya}, we get that any such limit point is a stationary distribution. Thus if (\ref{pro_lim}) converges, $\{E_k\}$ becomes an AMS sequence. If $\{E_k\}$ is Harris recurrent (\cite{Athreya}), then if the stationary distribution exists, it is unique. Thus, Harris recurrence and Feller continuity imply that $\{E_k\}$ is AMS  and (\ref{pro_lim})  converges to its unique stationary distribution.

If distribution of $Y_k$ or $T_k$ has a component which is absolutely continuous on any interval of $\mathbb{R}^+$, then $\{E_k\}$ is Harris recurrent with respect to the Lebesgue measure on $[0,\Gamma]$. More generally, if $\{E_k\}$ is not Harris irreducible with respect to a probability measure, since $\{E_k\}$ is a T-chain (\cite{Mey_Tweed}), its state space $[0,\Gamma]$ can be partitioned into atmost a countable number of sets $\{H_i,~i\in \mathcal{I}\}$ and $A$ where $H_i$ are Harris sets (i.e., $H_i$ are absorbing sets such that $\{E_k\}$ restricted to $H_i$ is Harris recurrent). Since the state space is compact, eventually, $\{E_k\}$ will get absorbed into one of the Harris sets and hence $\{E_k\}$ is AMS for any initial distribution.

When $\{E_k\}$ has a single Harris set then there is a unique stationary distribution and $\{E_k\}$ is AMS, ergodic. Otherwise, $\{E_k\}$ is AMS but has multiple ergodic components and depending on the initial conditions $\{E_k\}$ may or may not be AMS ergodic. We will assume that the system can be started  with any initial battery charge $E_0 \in [0,~\Gamma]$. This can be easily ensured in practice, by charging the battery initially to the level desired and then starting the system. This can ensure that if $\{E_k\}$ has more than one ergodic classes, then we can start $E_0$ in a state such that the chain will operate in only one specific ergodic class (e.g., class which will provide maximum mutual information) and then we will obtain $\{E_k\}$ which is AMS ergodic.

Thus, we will from now on assume that no matter which Markov energy management policy we use, we will obtain AMS ergodic sequence $\{(E_k,\hat{E}_k,X_k,Y_k,W_k)\}$. 
\section{Capacity of Finite Buffer System with BSIR}
In this section, we provide the capacity $C$ of the system. We assume that $\hat{E}_k$ is known to the transmitter and the receiver at time $k$. We now give the main result of the section.

 We will say that a rate $R$ is achievable if for each $n$ there exists an ($2^{nR},~n$) encoder at the transmitting user and a decoder such that the average probability of error $P^{(n)}_e \rightarrow 0$ as $n \rightarrow \infty$. $\mathcal{M} \triangleq [1:2^{nR}]$ denotes the message set to be transmitted. Capacity $C$ is the limit of achievable rates. As in \cite{Robert_Gray}, we will limit ourselves AMS, ergodic capacities.
\begin{theorem}
\label{T1} 
The capacity of a finite buffer energy harvesting system with buffer size $\Gamma$ is given by
\begin{flalign}
\label{exp_pi}
\hspace{40pt}C= \sup_{\pi,~P_{X(s)}}\sum_s \pi_s I(X(s);W).
\end{flalign}
where the stationary distributions $\pi $ are obtained via the Markov policies and $P_{X(s)}$ is any distribution on $X$ with $X^2 \leq s$.
\end{theorem}
\begin{proof}
See Appendix A.
\end{proof}
In the above theorem, if a Markov policy provides $\{E_k\}$ with more than one stationary distribution (i.e., $\{E_k\}$ has more than one ergodic sets) then, each of them is separately used to compute the RHS of (\ref{exp_pi}), i.e., in that case we will start with an initial state in the ergodic set which provides the largest RHS in (\ref{exp_pi}).\\
\textbf{Remarks:} The optimum policy does not permit any closed form expression and has to be numerically computed. The numerical computation is easy when the state space is finite. When the state space is uncountable, it may be difficult to compute the capacity. In the next section we provide a tight achievable lower bound for this case which is easier to compute. 
\section{Achievable Rate via Truncated Gaussian Signalling}
As stated in the previous section, the computation of the channel capacity especially when the state space is uncountable is a difficult task. Hence we look for easily computable achievable rates that  well approximate the finite buffer channel capacity under certain conditions. Truncated Gaussian signalling, which is the capacity achieving signalling scheme (\cite{Rajesh1}) for infinite buffer model, is one such scheme for the uncountable case. We use it in the following and compare with the capacity. Let $\{X_k',~k\geq 0\}$ be an i.i.d sequence with distribution  $\mathcal{N}(0,\mathbb{E}[Y]-\epsilon)$ where $\epsilon > 0$ is a small constant.

We consider the truncated Gaussian policy,
\begin{align}
\label{trunc_gauss_eqn}
\hspace{30pt} X_{k}=\text{sgn}(X_{k}')\min\{\sqrt{E_k+Y_{k}},|X_{k}'|\}.
\end{align}
where sgn($x$)$ = x/|x|$ for $x\neq 0$ and sgn$(x)=0,$ for $x=0$. Thus the energy buffer evolves as
\begin{align}
\hspace{30pt} E_{k+1}=\min\{\Gamma, \max\{0, E_{k}+\eta_{k+1}\}\}.
\end{align}
where $\eta_{k} \triangleq Y_k-X_k'^2$ for $k\geq 1$. We have $\mathbb{E}[{\eta_k}]=\epsilon >0.$ 

Let $R(\Gamma)$ denote the rate achieved via truncated Gaussian signalling for buffer length $\Gamma$. Now, we will denote the corresponding capacity by $C(\Gamma)$ and $E_k$ and $X_k$ by $E_k(\Gamma)$ and $X_k(\Gamma)$. We have that $R(\Gamma)$ $\leq$ $C(\Gamma)$ $\leq$ $C(\infty)$ $\triangleq$ $\frac{1}{2}\log(1+\frac{\mathbb{E}[Y]}{\sigma^2})$ where $C(\infty)$ was obtained in (\cite{Rajesh1}, \cite{Ulukus1}). We will show that $R(\Gamma){\rightarrow} C(\infty)$ as $\Gamma \rightarrow \infty$. With the energy management policy (\ref{trunc_gauss_eqn}), $\{E_k(\Gamma),~k\geq 0\}$ is a Markov chain. 
\begin{theorem}
\label{T2} 
The following holds:
\begin{flalign}
\label{lim_C}
\hspace{30pt} \lim_{\Gamma \rightarrow \infty}R(\Gamma)=\frac{1}{2}\log\left(1+\frac{\mathbb{E}[Y]}{\sigma^2}\right) .
\end{flalign}
\end{theorem}
\begin{proof}
See Appendix B.
\end{proof}
\subsection*{Example : Uncountable State Space}
Now $Y_k$ and $E_k$ can take any non-negative value with $E_k \leq \Gamma$. We take $\Gamma = 4$. $Y_k$ is taken to be uniformly distributed over $[0,Y_{\max}]$ where $Y_{\max}$ is varied to obtain various $\mathbb{E}[Y]$ values and $\sigma^2$ is chosen to be 1. The capacity of such a system is compared against the rate achieved via truncated Gaussian signalling in Fig \ref{cont_comp}. The corresponding $C(\infty)$ is also plotted. We see that the truncated Gaussian signalling provides a very good approximation to the capacity even though the buffer size is quite small and $C(\Gamma)$ is not close to $C(\infty)$. 
 
 In Fig \ref{plot_gamma}, we also plot the rate achieved for truncated Gaussian case, as a function of $\Gamma$. The rate is computed via algorithm provided in \cite{loeliger2}. 
\subsection{Finite State Space}
Now we consider the where the energy is quantized and $X_k^2$, $E_k$ and $Y_k$ take non negative integer values in terms of the energy quantizer. Now $E_k$ is a finite state Markov chain. In this case computing capacity (\ref{exp_pi}) is not that difficult. Also truncated Gaussian signalling (\ref{trunc_gauss_eqn}) cannot be used. However, we can still define a greedy policy. In greedy policy we  generate the channel input symbol $X_k$, with  the distribution optimizing $I(X(s);W)$ corresponding to the peak power $\hat{E}_k=s$, optimized over the finite state space. The optimization is done via the steepest descent algorithm. 
\subsubsection*{Example}
 We fix buffer size $\Gamma = 4$. $Y_k$ is taken to be uniformly distributed over $\{0,1\hdots Y_{\max}\}$ and $Y_{\max}$ is varied to obtain different $\mathbb{E}[Y]$. $\sigma^2 $ is taken to be 1. We plot the capacity $C(\Gamma)$ as a function of average harvested energy $\mathbb{E}[Y]$. The capacity is compared with the rate achieved via the greedy policy. In Fig \ref{disc_comp}, we plot the capacity via (\ref{exp_pi}), the achievable rate via the greedy policy and the capacity of the infinite buffer case. One sees that the greedy policy provides a rate close to the capacity.
 \begin{figure}[h]
 \hspace{40pt}
\includegraphics[scale=.45]{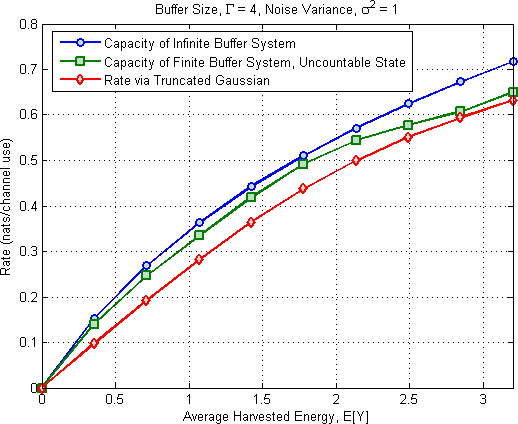}
\hspace{40pt}\caption{Finite buffer system capacity for uncountable state space.}
\label{cont_comp}
\end{figure}
 \begin{figure}[h]
 \hspace{40pt}
\includegraphics[scale=.45]{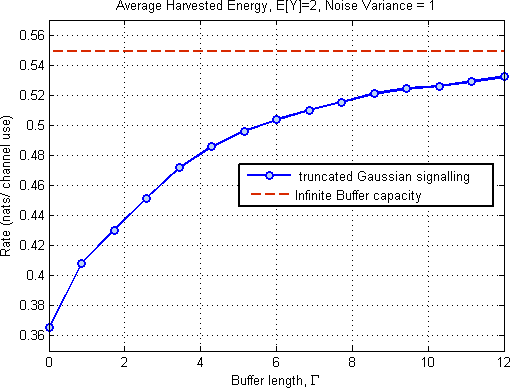}
\hspace{40pt}\caption{Convergence of rate obtained via truncated Gaussian signalling to the infinite buffer capacity.}
\label{plot_gamma}
\end{figure}
\begin{figure}[h]
 \hspace{40pt}
\includegraphics[scale=.45]{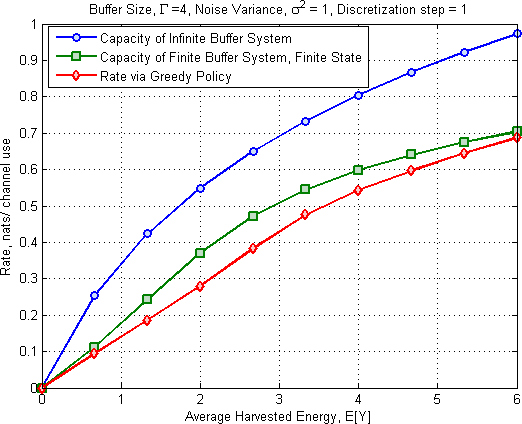}
\hspace{40pt}\caption{Finite buffer system capacity for finite state space.}
\label{disc_comp}
\end{figure}
\section{Capacity of Finite Buffer System with No BSIR}
In this section, we provide the capacity of a finite buffer system when $\hat{E}_k$ is not known at the receiver. This corresponds to an AWGN channel with  no BSIR. For the particular case of i.i.d states and finite input, output and state alphabet, \cite{Shannon58} provides the capacity of a point-to-point channel. The capacity is expressed in terms of functions defined from the state space to the input alphabet (``Shannon strategies"). An extension of the result in \cite{Shannon58} to point-to-point channels with the input, the output and the state alphabet being the real line is provided in \cite{Bar_Wor}. A generalization of the channel model in \cite{Shannon58} to finite alphabet channels with the state selected by an indecomposable, aperiodic Markov chain is provided in \cite{Jelinek}.  

We construct a sequence of order $m$ channels without side information associated with the original channel model as in \cite{Jelinek}. Capacity of each of these channels is known. Capacity of the system under consideration is then expressed as the limit of the capacity of the order $m$ channels as $m \rightarrow \infty$. As in previous sections, for the channel model under consideration, the side information $\hat{E}_k$ is a Markov chain. Let $\{U_k\}$ denote i.i.d auxiliary random variables. Let $f^{(m)}: \hat{\mathcal{E}}^{(m)} \rightarrow \mathcal{X}$ be functions defined on the $m-$tuple state space of $\{\hat{E}\}$ to the input alphabet (either finite or uncountable) called as ``strategy letters" in \cite{Jelinek}. 

 The order $m$ channel without side information is shown in Fig \ref{No_CSIR}. First we observe that the capacity result in \cite{Jelinek} holds for the case of uncountable alphabet as well. The capacity $C(m)$ of the channel model depicted in Fig \ref{No_CSIR},  is obtained as a direct extension of Shannon strategies to the uncountable alphabet case as given in \cite{Bar_Wor}. Then, the AWGN channel capacity of energy harvesting transmitters with a finite energy buffer and the energy available for transmission not known at the receiver is expressed as the limiting value of $C(m)$ as $m \rightarrow \infty$ (\cite{Jelinek}). The Markov policies which do not provide Harris recurrent Markov chains can be handled as in Section II. 
\begin{figure}[h]
 \hspace{-2pt}
\includegraphics[scale=.5]{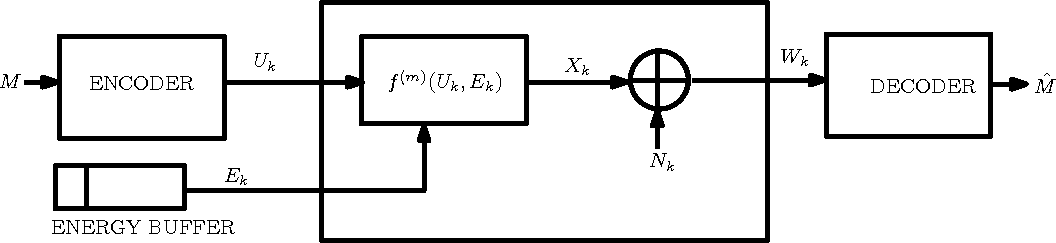}
\hspace{-100pt}
\hspace{40pt}\caption{Order $m$ channel without side information.}
\label{No_CSIR}
\end{figure}
 \begin{align*}
 \end{align*}
\section{Conclusions}
In this paper, we have considered an AWGN channel with a transmitter powered by an energy harvesting source and is equipped with a finite energy buffer. We model it as a channel with side information. We find the capacity when the state information is causally available at the transmitter but may or may not be available at the receiver. We also find easily computable achievable rates by using truncated Gaussian signalling and prove the convergence of the finite buffer capacity to the infinite buffer capacity as buffer length tends to infinity. Numerical examples are provided to substantiate the theoretical results.
\section*{Appendix A}
\section*{Proof of Theorem 1}
\textbf{Achievability}:\\
Coding and Signalling Scheme:\\
We start with a state $E_0$ and  generate i.i.d sequence $\{X_k(s),~k\geq 0\}$ for each state $\hat{E}=s$ with a distribution such that $X_k^2(s) \leq s $ and the RHS of (\ref{exp_pi}) is maximized. This way, we form $2^{nR}$ codewords independent of each other for each $M$. Whenever $\hat{E}_k = s$ we will pick the next untransmitted $\{X_k(s),~k\geq 0\}$ from the corresponding codeword. \\
\hspace{1pt}As explained in Section II, $\{(\hat{E}_k,X_k,W_k),~k=1,\hdots n\}$ is an AMS, ergodic sequence at transmission of each message. Also, since the channel is memoryless, if $I(X_0;W_0) < \infty$, AEP (Asymptotic Equipartition Property) holds (\cite{Robert_Gray}). But $I(X_0;W_0) \leq H(X_0)$ where $H(X_0)$ is defined as in \cite{Robert_Gray}. Since $\mathbb{E}[X_0^2] < \infty$, we can show $H(X_0)< \infty$ and hence $I(X_0;W_0)<\infty$. \\    
Decoding:\\
The decoder adopts jointly typical decoding. \\
For this coding, decoding scheme,  $P_e^{(n)} \rightarrow 0$ if  $ R < \bar{I}(X;W)=\overline{\lim}_{n \rightarrow \infty} \frac{1}{n}I(X^n;W^n)$. Let the limiting stationary distribution for this sequence be denoted by $\bar{P}$ and its limiting rate be denoted by $\bar{I}(\bar{X};\bar{W})$. Let the corresponding Pinsker rates (\cite{Robert_Gray}) for the two sequences be $I^*(X;W)$ and $I^*(\bar{X};\bar{W})$. Then $\bar{I}(X;W) \geq I^*(X;W)=I^*(\bar{X};\bar{W})\leq \bar{I}(\bar{X};\bar{W}) $. We would like to show that the rate $\bar{I}(\bar{X};\bar{W})$ is achievable. For this it is sufficient to show that $I^*(\bar{X};\bar{W})= \bar{I}(\bar{X};\bar{W}) $. From Theorem 6.4.3, \cite{Robert_Gray}, if we show that $I(\bar{X}_1;\bar{X}_{-1},\hdots|\bar{X}_0)<\infty$ and $I(\bar{X}^K;\bar{W})<\infty$ where $\bar{X}^K=(\bar{X}_1,\hdots \bar{X}_{K})$ and $\bar{W}$ is the two sided stationary version of $\{\bar{W}_k,~k\geq 0\}$, then $I^*(\bar{X};\bar{W})=\bar{I}(\bar{X};\bar{W})$. However, 
\begin{flalign*}
&I(\bar{X}_1;\bar{X}_{-1},\hdots|\bar{X}_0)\stackrel{(a)}{\leq} H(\bar{X}_1|\bar{X}_0)\stackrel{(b)}{\leq} H(\bar{X}_1) \stackrel{(c)}{<} \infty,\\
&I(\bar{X}^K;\bar{W})\stackrel{(d)}{\leq} H(\bar{X}^K)\stackrel{(e)}{\leq} \sum_{k=1}^{K}H(\bar{X}_k) \stackrel{(f)}{<} \infty ,
\end{flalign*}
where the entropy $H(.)$ is as defined in (Chapter 7, \cite{Robert_Gray}).  $(a)$ follows from (Corollary 7.11, \cite{Robert_Gray}) and $(b)$ follows from the non-negativity of mutual information. $(c)$ follows from the facts that $\mathbb{E}[\bar{X}_1^2]<\infty$ and that for a discrete random variable $X$ (obtained by quantizing $X$) with $\mathbb{E}[X^2]<\infty$, the entropy $H(X)<\infty$. Along similar lines, $(d)$, $(e)$ and $(f)$ follow.  \\
\textbf{Converse}: Let there be AMS, ergodic $(2^{nR},~n)$ encoding, decoding scheme with average probability of error $P_e^{(n)} \rightarrow 0$ as $n \rightarrow \infty$ that satisfies the energy constraints. Let $M \in \{1,\hdots,~2^{nR}\}$ uniformly distributed messages need to be transmitted by this code book. Let $X^{n} \triangleq (X_1,\hdots,~X_n)$ be the corresponding codeword that is transmitted. Then,
\begin{align*}
R=\frac{1}{n}H(M)&=\frac{1}{n}H(M|W^n,~\hat{E}^n)+\frac{1}{n}I(M;W^n,~\hat{E}^n)\\
&=\frac{1}{n}H(M|W^n,~\hat{E}^n)+\frac{1}{n}I(M;W^n|\hat{E}^n)\\
&\stackrel{(a)}{\leq} P_e^{(n)}R+\frac{1}{n}I(X^n;W^n|\hat{E}^n)
\end{align*}
\begin{align*}
&=P_e^{(n)}R+\frac{1}{n}H(W^n|\hat{E}^n)\\
&~~~-\frac{1}{n}H(W^n|X^n,~\hat{E}^n)\\
&=P_e^{(n)}R+\frac{1}{n}\sum_{k=1}^{n}[H(W_k|W^{k-1},~\hat{E}^n)\\
&~~~-H(W_k|X_k,~\hat{E}_k)]\\
&\leq P_e^{(n)}R+\frac{1}{n}\sum_{k=1}^{n}H(W_k|\hat{E}_k)\\
&~~~-H(W_k|X_k,~\hat{E}_k)\\
&=P_e^{(n)}R+\frac{1}{n}\sum_{k=1}^{n}I(W_k;X_k|\hat{E}_k)\\
\end{align*}
where (a) follows via Fano's inequality. Thus as $n \rightarrow \infty $, we obtain $R\leq \sup_{\pi,~P_X(s)}\sum_s \pi_s I(X(s);W)$ with $X^2(s) \leq s$.
$~~~~~~~~~~~~~~~~~~~~~~~~~~~~~~~~~~~~~~~~~~~~~~~~~~~~~~~~~~~~~~~~~~~~~~~\square$
\noindent
\section*{Appendix B}
\section*{Proof of Theorem 2}
First we prove the existence of a unique stationary distribution for the Markov chain $\{E_k(\Gamma),~k\geq 0\}$ defined in (\ref{trunc_gauss_eqn}). Observe that $|\eta_k^-| \leq X_k'^2$, where $\eta_k^-=\min\{\eta_k,~0\}$ and $\eta_k$ is defined in Section IV . Since $X_k' \sim $  i.i.d, $\mathcal{N}(0,\mathbb{E}[Y]-\epsilon)$,
 \begin{align}
 \label{finite_moment}
\mathbb{E}[(\eta_k^-)^\alpha] \leq \mathbb{E}[|X_k'|^{2\alpha}] < \infty, ~\forall \alpha > 0.
\end{align}
 Let $ \bar{E}_{k} \triangleq E_{k}(\infty),~\tau'(\Gamma) \triangleq \inf\{k: \bar{E_k} \geq \Gamma\}$ and $\tau({\Gamma}) \triangleq \inf\{k: E_k(\Gamma) = \Gamma\}$. Assume $E_0(\Gamma) = \bar{E}_0= e \leq \Gamma$.  Then $\tau(\Gamma)=\tau'(\Gamma)$ a.s. . Also, from (\ref{finite_moment}), $~\mathbb{E}[({\tau^\alpha(\Gamma)}|E_0=e]\leq \mathbb{E}[(\tau^\alpha(\Gamma))|E_0=0] \leq \mathbb{E}[(\tau'^\alpha(\Gamma)|E_0=0] < \infty$ (\cite{Alan_Gut} , Chapter 3). The epochs when $E_k(\Gamma)=\Gamma$ corresponds to the regeneration epochs of the Markov chain $\{E_k(\Gamma)\}$ and hence there  exists a unique stationary distribution $\pi_E(\Gamma)$ for the Markov chain.
 
Now we show that $\pi_E(\Gamma)$ is stochastically non-decreasing in $\Gamma$. Let  $E_0(\Gamma)=E_0(\Gamma+1)=0$. We have $P(E_1(\Gamma) \geq y|E_0(\Gamma)=e_0) \leq P(E_1(\Gamma+1)\geq y|E_0(\Gamma+1)=e_0),$ for all $y$ and hence in particular for $e_0=0$. Denoting the weak limit of $ E_k(\Gamma)$ by $E(\Gamma)$, it follows that $P(E(\Gamma)\geq y) \leq P(E(\Gamma+1) \geq y)$ (\cite{Stoyan}).  We thus obtain,
\begin{align*}
\hspace{20pt}E_k(\Gamma) \stackrel{st}{\leq} E_k(\Gamma+1)~\text{and}~ {E}(\Gamma) \stackrel{st}{\leq} E(\Gamma+1),
\end{align*}
where $X \stackrel{st}{\leq}Y$ denotes $P(X>x)\leq P(Y>x)$ for all $x$.
Let $E_0(\Gamma)=E_0=0$. Then, for any $k \geq 0,$ 
\begin{align}
\label{st_convrg}
\hspace{30pt}E_k(\Gamma) \stackrel{st}{\nearrow}\bar{ E}_k ~\text{as}~ \Gamma \rightarrow \infty.
\end{align}
We also know that, $0=\bar{E}_0 \stackrel{st}{\leq} \hdots \stackrel{st}{\leq} \bar{E}_k $ converges weakly to $\infty$ as $k \rightarrow \infty$. Thus $\lim_{k \rightarrow \infty}P(E_k \geq x)=1$ for any $x$. 
By stochastic monotonicity in $k$ (when $E_0(\Gamma)=0$) and $\Gamma$,
\begin{align*}
\hspace{20pt} P(E_k(\Gamma) \geq x) \geq 1-\epsilon,  ~ \forall~ \Gamma ~\geq ~\Gamma_1,~k \geq~ N_1,
\end{align*}
for some $\Gamma_1$ and $N_1$. Again by stochastic monotonicity in $k$ for $\Gamma \geq \Gamma_1$,  
\begin{align*}
\hspace{30pt} P(E(\Gamma) \geq x) \geq 1-\epsilon, ~ \forall ~\Gamma \geq \Gamma_1.
\end{align*}
Therefore,
\begin{align*}
&\lim_{\Gamma \rightarrow \infty}P_{\pi_\Gamma}(|X_k(\Gamma)|= |X_k'|)\\
&=\lim_{\Gamma \rightarrow \infty} P_{\pi_\Gamma}(\sqrt{E_{k-1}(\Gamma)+Y_k} \geq |X_k'|)=1.
\end{align*}
Also, $\mathbb{E}[X_k^2(\Gamma)] \leq \mathbb{E}[X_k'^2]= \mathbb{E}[Y]-\epsilon.$ Therefore, we obtain (\cite{Ver_fun_mmse}, Theorem 9), as $\Gamma \rightarrow \infty$,  
\begin{align*}
&\bar{I}(X_k(\Gamma);W_k) \rightarrow \bar{I}(X_k(\infty);W_k(\infty))=\\
&\frac{1}{2}\log\left(1+\frac{\mathbb{E}[Y]-\epsilon}{\sigma^2}\right).
\end{align*}
 Hence,
\begin{align*}
 \frac{1}{2}\log(1+\frac{\mathbb{E}[Y]}{\sigma^2}) &\geq \lim_{\Gamma \rightarrow \infty} \bar{I}(X_k(\Gamma);W_k)\\
 &\geq \frac{1}{2} \log(1+\mathbb{E}[Y]-\epsilon), ~ \forall~\epsilon > 0.
\end{align*}
We take $\epsilon \rightarrow 0$  and get the required result. $~~~~~~~~~~~~~~~~~~~~~~\square$

\bibliographystyle{IEEEtran}
\bibliography{fb_cap_fin_16713}
 \end{document}